\documentclass[twoside,11pt]{entics} 
\usepackage{enticsmacro}
\usepackage[all]{xy}
\usepackage[utf8]{inputenc}
\usepackage[T1]{fontenc}
\usepackage{microtype}
\usepackage{graphicx}
\usepackage{amsmath}
\usepackage{stmaryrd}
\usepackage{enumitem}
\usepackage{tikz,tikz-cd}
\usepackage{newunicodechar}
\newunicodechar{∋}{\ni}
\newunicodechar{∙}{\cdot}
\newunicodechar{·}{\cdot}
\newunicodechar{⊢}{\vdash}
\newunicodechar{⋆}{{}^\star}
\newunicodechar{Π}{\Pi}
\newunicodechar{⇒}{\Rightarrow}
\newunicodechar{ƛ}{\lambdabar}
\newunicodechar{∅}{\emptyset}
\newunicodechar{∀}{\forall}
\newunicodechar{∃}{\exists}
\newunicodechar{ϕ}{\Phi}
\newunicodechar{ψ}{\Psi}
\newunicodechar{ρ}{\rho}
\newunicodechar{α}{\alpha}
\newunicodechar{β}{\beta}
\newunicodechar{δ}{\delta}
\newunicodechar{π}{\pi}
\newunicodechar{τ}{\tau}
\newunicodechar{ι}{\iota}
\newunicodechar{μ}{\mu}
\newunicodechar{σ}{\sigma}
\newunicodechar{≡}{\equiv}
\newunicodechar{Γ}{\Gamma}
\newunicodechar{∥}{\parallel}
\newunicodechar{Λ}{\Lambda}
\newunicodechar{θ}{\theta}
\newunicodechar{Θ}{\Theta}
\newunicodechar{∘}{\circ}
\newunicodechar{Δ}{\Delta}
\newunicodechar{λ}{\lambda}
\newunicodechar{⊧}{\models}
\newunicodechar{⊎}{\uplus}
\newunicodechar{η}{\eta}
\newunicodechar{Σ}{\Sigma}
\newunicodechar{ξ}{\xi}
\newunicodechar{₁}{_1}
\newunicodechar{₂}{_2}
\newunicodechar{₃}{_3}
\newunicodechar{ₙ}{_n}
\newunicodechar{ᵢ}{_i}
\newunicodechar{ᴬ}{^A}
\newunicodechar{ᴮ}{^B}
\newunicodechar{ℕ}{\mathbb{N}}
\newunicodechar{ᶜ}{{}^c}
\newunicodechar{Φ}{\Phi}
\newunicodechar{Ψ}{\Psi}
\newunicodechar{≐}{\doteq}
\newunicodechar{→}{\rightarrow}
\newunicodechar{←}{\leftarrow}
\newunicodechar{×}{\times}
\newunicodechar{⁻}{^{-}}
\newunicodechar{¹}{^1}
\newunicodechar{∗}{\ast}
\newunicodechar{≡}{\equiv}
\newunicodechar{⊥}{\bot}
\newunicodechar{⊤}{\top}
\newunicodechar{⋯}{\dots}

\makeatletter
\def\ordinarycolon{\mathchar`\:}
\def\colon{\,\mathord\ordinarycolon\,}

\protected\edef\tikz@nonactivecolon{%
  \noexpand\ifmmode\noexpand\colon\noexpand\else:\noexpand\fi
}
\begingroup\lccode`\~=`\:\lowercase{\endgroup
  \protected\def~}{\new@ifnextchar={\coloneqq\@gobble}{\colon}}
\AtBeginDocument{\mathcode`\:="8000 }
\makeatother

\let\bar\overline
\let\hat\widehat

\newcommand*{\unary}[2]{{#1}\mkern\thinmuskip{#2}}
\newcommand*{\binary}[3]{\unary{\unary{#1}{#2}}{#3}}
\newcommand*{\ternary}[4]{\unary{\binary{#1}{#2}{#3}}{#4}}
\newcommand*{\quaternary}[5]{\unary{\ternary{#1}{#2}{#3}{#4}}{#5}}
\newcommand*{\quinary}[6]{\unary{\quaternary{#1}{#2}{#3}{#4}{#5}}{#6}}
\newcommand*{\senary}[7]{\unary{\quinary{#1}{#2}{#3}{#4}{#5}{#6}}{#7}}
\newcommand*{\septenary}[8]{\unary{\senary{#1}{#2}{#3}{#4}{#5}{#6}{#7}}{#8}}

\DeclareMathOperator{\Setoperator}{\sf Set}
\newcommand*{\Set}{\Setoperator}
\DeclareMathOperator{\Reloperator}{\sf Rel}
\newcommand*{\Rel}[2]{\Reloperator\mkern\thinmuskip{#1}\mkern\thinmuskip{#2}}
\DeclareMathOperator{\dataoperator}{\sf data}
\DeclareMathOperator{\whereoperator}{\sf where}
\newcommand*{\data}[1]{\binary{\dataoperator}{#1}{\whereoperator}}
\newcommand*{\var}{\mathsf}

\DeclareMathOperator{\Listoperator}{\sf List}
\newcommand*{\List}{\unary{\Listoperator}}
\DeclareMathOperator{\niloperator}{\sf nil}
\newcommand*{\nil}{\niloperator}
\DeclareMathOperator{\consoperator}{\sf cons}
\newcommand*{\cons}{\binary{\consoperator}}
\DeclareMathOperator{\Listliftoperator}{\hat{\Listoperator}}
\newcommand*{\Listlift}{\ternary{\Listliftoperator}}
\DeclareMathOperator{\nilliftoperator}{\hat{\niloperator}}
\newcommand*{\nillift}{\nilliftoperator}
\DeclareMathOperator{\consliftoperator}{\hat{\consoperator}}
\newcommand*{\conslift}{\septenary{\consliftoperator}}

\DeclareMathOperator{\Seqoperator}{\sf Seq}
\newcommand*{\Seq}[1]{\Seqoperator\mkern\thinmuskip{#1}}
\DeclareMathOperator{\injoperator}{\sf inj}
\newcommand*{\inj}[1]{\injoperator\mkern\thinmuskip{#1}}
\DeclareMathOperator{\pairingoperator}{\sf pairing}
\newcommand*{\pairing}[2]{\pairingoperator {#1}\mkern\thinmuskip{#2}}
\DeclareMathOperator{\Seqliftoperator}{\hat{\Seqoperator}}
\newcommand*{\Seqlift}[3]{%
  \Seqliftoperator
  \mkern\thinmuskip{#1}%
  \mkern\thinmuskip{#2}%
  \mkern\thinmuskip{#3}%
}
\DeclareMathOperator{\injliftoperator}{\hat{\injoperator}}

\DeclareMathOperator{\pairingliftoperator}{\hat{\pairingoperator}}

\DeclareMathOperator{\Soperator}{\sf S}
\newcommand*{\Stype}[1]{\Soperator\mkern\thinmuskip{#1}}
\DeclareMathOperator{\ioperator}{\sf i}
\newcommand*{\icons}[1]{\ioperator\mkern\thinmuskip{#1}}
\DeclareMathOperator{\poperator}{\sf p}
\newcommand*{\pcons}[3]{\poperator
  \mkern\thinmuskip{#1}%
  \mkern\thinmuskip{#2}%
  \mkern\thinmuskip{#3}%
}
\DeclareMathOperator{\Sliftoperator}{\hat{\Soperator}}
\newcommand*{\Slift}[3]{%
  \Sliftoperator
  \mkern\thinmuskip{#1}%
  \mkern\thinmuskip{#2}%
  \mkern\thinmuskip{#3}%
}
\DeclareMathOperator{\iliftoperator}{\hat{\ioperator}}
\newcommand*{\ilift}[4]{%
  \iliftoperator%
  \mkern\thinmuskip{#1}%
  \mkern\thinmuskip{#2}%
  \mkern\thinmuskip{#3}%
  \mkern\thinmuskip{#4}%
}
\DeclareMathOperator{\pliftoperator}{\hat{\poperator}}

\DeclareMathOperator{\mapoperator}{\sf map}%
\newcommand*{\mapfun}[1]{%
  \mapoperator_{#1}%
}%
\newcommand*{\mapfunap}[2]{%
  \unary{\mapfun{#1}}{#2}%
}%
\newcommand*{\mapfunaparg}[3]{%
  \unary{\mapfun{#1}{#2}}{#3}%
}%

\DeclareMathOperator{\gmapoperator}{\sf m}%
\newcommand*{\gmapfun}[1]{%
  \gmapoperator_{#1}%
}%
\newcommand*{\gmapfunap}[2]{%
  \unary{\gmapfun{#1}}{#2}%
}%
\newcommand*{\gmapfunaparg}[3]{%
  \unary{\gmapfun{#1}{#2}}{#3}%
}%

\DeclareMathOperator{\Booloperator}{\sf Bool}
\newcommand*{\Bool}{\Booloperator}
\DeclareMathOperator{\falseoperator}{\sf false}
\newcommand*{\false}{\falseoperator}
\DeclareMathOperator{\trueoperator}{\sf true}
\newcommand*{\true}{\trueoperator}

\newcommand*{\rel}[3]{
  {#1}%
  \mkern\thinmuskip{#2}%
  \mkern\thinmuskip{#3}%
}

\newcommand*{\product}[2]{{#1}\mathbin{\sf \times}{#2}}
\newcommand*{\pair}[2]{({#1},{#2})}
\newcommand*{\productlift}[2]{{#1}\mathbin{\sf \hat\times}{#2}}

\newcommand*{\arrow}[2]{{#1}\mathbin{\sf →}{#2}}
\newcommand*{\id}[1]{{\sf id}_{#1}}
\newcommand*{\arrowlift}[2]{{#1}\mathbin{\sf \hat{→}}{#2}}

\newcommand*{\Graph}{\unary{\sf Gr}}
\newcommand*{\Graphap}{\ternary{\sf Gr}}

\def\conf{MFPS 2024}
\volume{NN}

\def\lastname{Cagne, Johann}

\begin{document}

\begin{frontmatter}
  \title{Early Announcement: Parametricity for GADTs}

  \author{Pierre Cagne\thanksref{a}\thanksref{cagnep}}
  \author{Patricia Johann\thanksref{a}\thanksref{johannp}}
  
  \address[a]{Department of Computer Science\\ Appalachian State University\\
    Boone, NC, USA}

  \thanks[cagnep]{Email: \href{mailto:cagnep@appstate.edu}
    {\texttt{\normalshape cagnep@appstate.edu}}}
  
  \thanks[johannp]{Email: \href{mailto:johannp@appstate.edu}
    {\texttt{\normalshape johannp@appstate.edu}}}
  
\begin{abstract} 
  Relational parametricity was first introduced by Reynolds for System
  F. Although System F provides a strong model for the type systems at
  the core of modern functional programming languages, it lacks
  features of daily programming practice such as complex data
  types. In order to reason parametrically about such objects,
  Reynolds' seminal ideas need to be generalized to extensions of
  System F. Here, we explore such a generalization for the extension
  of System F by Generalized Algebraic Data Types (GADTs) as found in
  Haskell. Although GADTs generalize Algebraic Data Types (ADTs) ---
  i.e., simple recursive types such as lists, trees, etc. --- we show
  that naively extending the parametric treatment of these recursive
  types is not enough to tackle GADTs. We propose a tentative
  workaround for this issue, borrowing ideas from the categorical
  semantics of GADTs known as (functorial) completion. We discuss some
  applications, as well as some limitations, of this solution.
\end{abstract}

\end{frontmatter}


\section{Introduction}
\label{sec:intro}

Relational parametricity~\cite{rey83} is a key technique for reasoning
about programs in strongly typed languages. It can be used to enforce
invariants guaranteeing strong properties of programs, programming
languages, and programming language implementations supporting
parametric polymorphism. A polymorphic program is a program that can
be applied to arguments and return results of different types; a
parametric polymorphic program is a program that is not only
polymorphic over {\em all} types, but is also defined by the same
type-uniform algorithm regardless of the concrete type at which it is
applied. Since parametric polymorphic programs cannot perform
type-specific operations, the computational behaviors they can exhibit
are actually quite constrained. Parametricity was originally put forth
by Reynolds~\cite{rey83} for System F~\cite{gir72,rey74}, the formal
calculus at the core of all polymorphic functional languages. It was
later popularized as Wadler's ``theorems for free''~\cite{wad89},
so-called because it allows the deduction of properties of programs in
such languages solely from their types, i.e., with no knowledge
whatsoever of the text of the programs involved. However, to get
interesting free theorems, Wadler actually treats System F extended
with built-in lists. Indeed, most of the free theorems in~\cite{wad89}
are essentially naturality properties for polymorphic list-processing
functions. It is easy to extend the techniques developed
in~\cite{wad89} for handling lists to non-list algebraic data types
(ADTs). Parametricity for such types can then be used to derive not
just naturality (i.e., commutativity) properties, but also results ---
such as proofs of type inhabitance and correctness of the program
optimization known as {\em short cut fusion}~\cite{glp93} --- that go
beyond simple naturality.

In his original formulation, Reynolds gives each type expression of
System F a {\em relational interpretation} defined inductively. Each
type expression $Φ$ with type variables $α₁,α₂,⋯,αₙ$ thus gives, for
each tuple $\bar R$ of relations $Rᵢ$ between types $Aᵢ$ and $Bᵢ$, a
relation $\hat{Φ}\,\bar R$ between the type $Φ[\bar A/\bar \alpha]$ and
$Φ[\bar B/\bar \alpha]$. To capture the intended type-uniformity of
System F's polymorphic expressions, these relational interpretations
are defined in such a way that every function
$f : ∀ \bar \alpha. Φ → Ψ$, where $Φ$ and $Ψ$ are two type expressions
in the same type variables $\bar \alpha$, is {\em parametric} in the
following sense: for each tuple of relations $\bar R$, the pairs
related by $\hat {Φ}\,\bar R$ are sent by $f$ to pairs related by
$\hat {Ψ}\,\bar R$.

As mentioned above, better approximations of realistic programming
languages result from adding built-in data types to System F. Each
such added data type induces a type constructor, and this type
constructor must also be given a relational
interpretation. Wadler~\cite{wad89} considers the case of lists, which
we review in detail in Section~\ref{sec:inductive}. To add a new
inductive data type constructor $T$ to an ambient parametric language
in such a way that parametricity is preserved, the method is always
the same: Define its relational interpretation as a (dependent)
inductive family $\hat T$ with one data constructor $\hat c$ for each
data constructor $c$ of $T$ expressing precisely that $c$ is a
parametric polymorphic function. The data constructors of such a data
type's relational interpretation thus make formal the intuitive
type-uniformity required of its data constructors by the grammars of
languages such as Haskell. The relational interpretation $\hat T$
captures the intuition that, if we regard data types as containers,
then two data structures of (two instances of) $T$ are related by
$\hat T\, R$ exactly when the data they store are related by $R$. This
intuition also requires that $\hat T$ preserves inclusion, i.e., that
$\hat T\, R \subseteq \hat T\, S$ whenever $R \subseteq S$. Indeed, if
two data structures are related by $\hat T\, R$, then the data they
store are related by $R$, and thus by $S$, so the two data structures
must be related by $\hat T \, S$. Fortunately, for lists and other
ADTs, the relational interpretations defined in this way enjoy this
crucial inclusion-preservation property.\looseness=-1

Here, we report our ongoing efforts to add the generalization of ADTs
known as {\em Generalized Algebraic Data Types (GADTs)} to System F in
such a way that parametricity is preserved. In doing so, we insist on
understanding GADTs as types of {\em data structures}, i.e., as types
of containers that can be filled with data.  Since this entails in
particular that GADTs are inductive data type constructors, we might
expect that following the method outlined above will suffice. In
Section~\ref{sec:inductive}, we show that naively doing so results in
relational interpretations of GADTs that do not satisfy the
inclusion-preservation property identified at the end of the preceding
paragraph. This is problematic: if we are to understand GADTs as types
of data structures, then they should certainly satisfy all properties
--- among them the inclusion-preservation property --- expected of
such types.  In Section~\ref{sec:gadts-comp}, we explore a promising
approach to overcoming this issue. This approach consists in defining
the relational interpretation of a GADT through that of its {\em
  completion}, an ADT-like type constructor that contains the original
GADT. In Section~\ref{sec:app} we offer some applications of
parametricity for GADTs obtained using our proposed approach. In
Section~\ref{sec:issues} we discuss some issues that arise when making
our proposed approach precise. Doing so requires defining a source
language (an extension of System F that allows for GADTs), a target
language (a dependent type theory strong enough to encode relations),
and interpretations of each type of the source language as both a type
and a relation in the target language. We point out some difficulties
in the design of the target language, and also offer some thoughts on
how to resolve them. Throughout the paper, we use an Agda-like syntax
to write examples of types and terms of the anticipated target
language. We note, however, that this language might end up being very
different from Agda's type theory. In particular, this early
announcement by no means reports on an attempt to formalize our work
in a proof assistant.

We are not the first to consider parametricity for GADTs. Very recent
progress on the subject has been presented
in \cite{birkedal:essence}. Sieczkowski {\em et al.}~construct there a
parametric model of an extension of System F supporting GADTs, with
the aim of deriving free theorems and representation independence
results. However, their work differs drastically from the line of
research presented here in several ways. First, the semantics
presented by Sieczkowski {\em et al.}~targets
normalization-by-evaluation. By contrast, our work is in no way
concerned with such methods. Second, Sieczkowski et al. make essential
use of guarded recursion through a universe of step-indexed
propositions equipped with a later modality (as exists, e.g., in
Iris). By contrast, we are concerned only with structural recursion in
this work. Third, Sieczkowski {\em et al.}~insist on the importance of
two particular rules of their type system: discriminability and
injectivity of type constructors. By contrast, we are agnostic about
such rules, thus accommodating more diverse host languages. Finally,
and most importantly, the semantics of Sieczkowski {\em et al.}~models
parametricity for GADTs only in those type indices that are {\em
unconstrained}, i.e., that can be promoted to parameters. In
particular, their approach cannot handle free theorems such as the one
presented in Section~\ref{sec:free-thm} for $\Seqoperator$, since
$\pairingoperator$ has a constrained instance of $\Seqoperator$ as
return type. By contrast, we not only recognize the non-uniformity of
GADTs acknowledged by Sieczkowski {\em et al.}, but we also recognize
that this break of uniformity is governed by uniform type constructors
(namely, those constraining the instances of the return types of
GADTs' data constructors), and that this uniformity must be captured
by parametric models of the language at play.




\section{Naive approach: The problem}
\label{sec:inductive}

In this section, we first review Wadler's relational interpretation of
the standard built-in type constructor $\mathsf{List}$ for the ADT of
lists, and then try to extend the method directly to GADTs. As noted
in Section~\ref{sec:intro}, the resulting relational interpretations
for GADTs lack the desired inclusion-preservation
property.\looseness=-1

The type constructor $\mathsf{List}$ for the ADT of lists is given
by:
\begin{equation}\label{eq:lists}
\begin{array}{l}
  \data{\Listoperator : \Set → \Set}  \\
  \quad \niloperator : ∀\{\var{α}\} → \List {\var{α}}\\
  \quad \consoperator : ∀\{\var{α}\} → \var{α} → \List {\var{α}} → \List {\var{α}} 
\end{array}
\end{equation}
Wadler effectively gave a relational interpretation for
$\Listoperator$ informally when he declared two lists
$[a_1,a_2,...,a_n] : \List A$ and $[b_1,b_2,...,b_n] : \List B$ to be
related by $\unary{\Listliftoperator} R$ for a relation $R$ between
the types $A$ and $B$ exactly when each pair of their corresponding
elements is related by $R$, i.e., when he required
\begin{equation*}
   \Listlift R {[a_1,a_2,⋯,a_n]} {[b_1,b_2,⋯,b_n]}
  \quad \text{if and only if} \quad
  ∀ i = 1,⋯,n,\, \rel R {a_i} {b_i}
\end{equation*}
If we represent the type $\Rel A B$ of relations between $A$ and $B$
by the function type $A → B → \Set$, then we can formalize Wadler's
relational interpretation for lists as the (dependent) inductive
family represented
by:\looseness=-1
\begin{equation*}
  \label{eq:list-lift}%
  \begin{array}{l}
    \data {\Listliftoperator :
    ∀ \{\var{α}\, \var{β}\} → \Rel {\var{α}} {\var{β}} →
    \Rel {(\List{\var{α}})} {(\List {\var{β}})}} \\
    \quad
    \nilliftoperator : ∀ \{\var{α}\,\var{β}\} (\var R : \Rel {\var{α}}{\var{β}}) →
    \Listlift {\var R} {\nil} {\nil}\\
    \quad
    \consliftoperator : ∀ \{\var{α}\,\var{β}\} (\var R : \Rel{\var{α}}{\var{β}})
    (\var a : \var {α})(\var b : \var{β})(\var{as} : \List {\var{α}})
    (\var{bs} : \List {\var{β}}) → \\
    \qquad \rel{\var R} {\var a} {\var b} → \Listlift {\var R} {\var{as}} {\var{bs}} →
    \Listlift {\var R} {(\cons {\var a} {\var {as}})} {(\cons {\var b} {\var {bs}})}
  \end{array}
\end{equation*}
Notice that only terms both constructed from the same data constructor
can be related, and that the definition of $\Listliftoperator$ mimics
the recursive structure of $\Listoperator$'s data type
declaration. This ensures in particular that $\Listliftoperator$
preserves inclusions:
%
Given an inclusion $R \subseteq S$ witnessed by
$i: ∀ a\, b → \rel R a b → \rel S a b$, the inclusion
$\Listliftoperator {R} \subseteq \Listliftoperator {S}$ is witnessed
by the function
$\unary \Listliftoperator i: ∀ \ell\, \ell' → \Listlift R \ell {\ell'}
→ \Listlift S \ell {\ell'}$ mapping $\nillift R$ to $\nillift S$ and
$\conslift R {a_h} {b_h} {a_t} {b_t} {w_h} {w_t}$ to
$\conslift S {a_h} {b_h} {a_t} {b_t} {(\unary i {w_h})} {(\binary
  \Listliftoperator i {w_t})}$. %
Moreover, this definition has exactly the feature announced in
Section~\ref{sec:intro}, namely that $\nilliftoperator$ and
$\consliftoperator$ express that $\niloperator$ and $\consoperator$
are parametric, respectively.

The method we used to construct $\Listoperator$ can easily be extended
to other ADTs, or even to more general inductive definitions. This is
the approach explored by the authors of~\cite{bernardy} for generic
inductive families, which encompass, in particular, GADTs. Although
interesting in its own right, this approach fails to recognize GADTs
as data types, in the sense that their relational interpretations do
not necessarily preserve inclusion as is intuitively expected. To
illustrate the issue, consider the GADT of sequences given
by:\looseness=-1 
\begin{equation}
  \label{eq:seq}%
  \begin{array}{l}
    \data{\Seqoperator : \Set → \Set}\\
    \quad\injoperator : ∀ \{\var{α}\} → {\var{α}} → \Seq{\var{α}}\\ 
    \quad\pairingoperator : ∀ \{\var{α}₁\,\var{α}₂\} → 
    \Seq {\var{α}₁} → \Seq{\var{α}₂} → \Seq{(\product{\var{α}₁}{\var{α}₂})}
  \end{array}
\end{equation}
The same method used above yields the following relational
interpretation $\sf\hat{Seq}$ for $\sf Seq$:
\begin{equation*}%
  \label{eq:seq-lift}
  \begin{array}{l}
    \data{\Seqliftoperator : ∀ \{\var{α}\,\var{β}\} → \Rel {\var{α}}{\var{β}} →
    \Rel {(\Seq {\var{α}})} {(\Seq {\var{β}})}} \\
    \quad\injliftoperator : ∀\{\var{α}\,\var{β}\} (\var R : \Rel{\var{α}}{\var{β}})
    (\var a : \var{α})(\var b : \var{β}) → \rel{\var R} {\var a} {\var b} →
    \Seqlift {\var R} {(\inj {\var a})} {(\inj {\var b})}\\
    \quad\pairingliftoperator : ∀ \{\var{α}₁\,\var{α}₂\,\var{β}₁\,\var{β}₂\}
    (\var R₁ : \Rel {\var{α}₁}{\var{β}₁}) (\var R₂ : \Rel {\var{α}₂}{\var{β}₂}) → \\
    \qquad ∀(\var s₁ : \Seq{\var{α}₁})(\var s₂ : \Seq{\var{α}₂})
    (\var t₁ : \Seq{\var{β}₁})(\var t₂ : \Seq{\var{β}₂}) → \\
    \qquad \Seqlift {\var R₁} {\var s₁} {\var t₁} → \Seqlift {\var R₂} {\var s₂}{\var t₂}
    → \Seqlift {(\productlift{\var R₁}{\var R₂})}
    {(\pairing {\var s₁} {\var s₂})} {(\pairing {\var t₁} {\var t₂})}
  \end{array}
\end{equation*}
Here, $\productlift \_ \_$ is the relational interpretation of the
product type constructor $\product \_ \_$ defined on relations
$R₁ : \Rel {A₁} {B₁}$ and $R₂ : \Rel {A₂} {B₂}$ by
$\rel{(\productlift {R₁} {R₂})} {(a₁,a₂)} {(b₁,b₂)} = \product {(\rel
  {R₁} {a₁} {b₁})} {(\rel {R₂} {a₂} {b₂})}$ for all $a₁ : A₁$,
$a₂ : A₂$, $b₁ : B₁$, $b₂ : B₂$.

Now, assuming extensionality for relations (i.e., $R$ is equal to $S$
when they relate the same elements), if $R$ is the equality relation
on $\product \Bool \Bool$, where $\Bool$ is the type of booleans, then
$\Seqliftoperator\, R$ is the equality relation on
$\Seq {(\product \Bool \Bool)}$. On the other hand, if $S$ is the
binary relation on $\product \Bool \Bool$ with only
$\rel S {\pair \false \false} {\pair \true \true}$ uninhabited, then
$\Seqlift S {(\pairing {s₁} {s₂})} {(\pairing {t₁} {t₂})}$ is
uninhabited for any $s₁$, $s₂$, $t₁$, and $t₂$ because $S$ is not a
product of relations. However, $S$ contains $R$, so $\Seqlift S s t$
should be inhabited at least whenever $\Seqlift R s t$ is. That it is
not violates the inclusion-preservation property expected of
$\Seqliftoperator$.\looseness=-1


\section{Completing GADTs: Toward a solution}
\label{sec:gadts-comp}

To remedy the problem exposed in Section~\ref{sec:inductive}, we first
observe that we can obtain an alternative relational interpretation
for $\Seqoperator$ by first embedding it into (a data type that is
essentially) an ADT $\Soperator$ and then constructing the relational
interpretation of $\sf S$ as above. The data type $\sf S$ is given by:
\begin{equation}\label{eq:seq-alt}
  \begin{array}{l}
    \data{\Soperator : \Set → \Set}\\
    \quad\ioperator : ∀ \{\var{α}\} → {\var{α}} → \Stype{\var{α}}\\ 
    \quad\poperator : ∀ \{\var{α}₁\,\var{α}₂\,\var{α}\} →
    (\arrow{\product {\var{α}₁}{\var{α}₂}}{\var{α}}) →
    \Stype {\var{α}₁} → \Stype{\var{α}₂} → \Stype{\var{α}}
  \end{array}
\end{equation}
We can compute the relational interpretation of $\Soperator $ as we
did above for ADTs. This gives:
\begin{equation*}\label{eq:seq-lift}
  \begin{array}{l}
    \data{\Sliftoperator : ∀ \{\var{α}\,\var{β}\} → \Rel {\var{α}}{\var{β}} →
    \Rel {(\Stype {\var{α}})} {(\Stype {\var{β}})}} \\
    \quad\iliftoperator : ∀\{\var{α}\,\var{β}\} (\var R : \Rel{\var{α}}{\var{β}})
    (\var a : \var{α})(\var b : \var{β}) → \rel{\var R} {\var a} {\var b} →
    \Slift {\var R} {(\icons {\var a})} {(\icons {\var b})}\\
    \quad\pliftoperator : ∀ \{\var{α}₁\,\var{α}₂\,\var{β}₁\,\var{β}₂\,\var{α}\,\var{β}\}
    (\var R₁ : \Rel {\var{α}₁}{\var{β}₁}) (\var R₂ : \Rel {\var{α}₂}{\var{β}₂})
    (\var R : \Rel {\var{α}}{\var{β}}) → \\
    \qquad ∀(\var s₁ : \Stype {\var{α}₁})(\var s₂ : \Stype {\var{α}₂})
    (\var t₁ : \Stype {\var{β}₁})(\var t₂ : \Stype {\var{β}₂})
    (\var f : \arrow {\product {\var{α}₁}{\var{α}₂}} {\var{α}})
    (\var g: \arrow {\product{\var{β}₁}{\var{β}₂}} {\var{β}})→ \\
    \qquad \rel{(\arrowlift {(\productlift{\var R₁}{\var R₂})} {\var R})} {\var f} {\var g} →
    \Slift {\var R₁} {\var s₁} {\var t₁} → \Slift {\var R₂} {\var s₂}{\var t₂}
    → \Slift {\var R}
    {(\pcons {\var f} {\var s₁} {\var s₂})} {(\pcons {\var g} {\var t₁} {\var t₂})}
  \end{array}
\end{equation*}
Here, $\arrowlift\_\_$ is the relational interpretation of the
function type constructor $\arrow\_\_$ defined for any $R:\Rel A B$
and $S : \Rel C D$ by
$\rel{(\arrowlift R S)} f g = ∀ (a : A)(b : B) → \rel R a b → \rel S
{(\unary f a)} {(\unary g b)}$.

Note that there is an embedding $\iota$ of $\Seqoperator$ into
$\Soperator$, obtained by mapping a sequence of the form $\inj a$ to
$\icons a$ and one of the form $\pairing {s₁} {s₂}$ with $s₁ : \Seq
{A₁}$ and $s₂ : \Seq {A₂}$ to $\pcons {\id {\product {A₁}{A₂}}}
{(\unary {ι} s₁)} {(\unary {ι} s₂)}$. With $ι$ and $\Sliftoperator$ in
hand, we can now define the relational interpretation of
$\Seqoperator$ by $\Seqlift R s t = \Slift R {(\unary \iota
s)}{(\unary \iota t)}$. It is easy to see that the relational
interpretation $\Seqliftoperator$ not only ensures that the
constructors of $\Seqoperator$ are parametric, but also satisfies the
inclusion-preservation property. Indeed, if $i : R \subseteq S$ then
we can define $\unary \Sliftoperator i : \unary \Sliftoperator
R \subseteq \unary
\Sliftoperator S$ simply by mapping $\ilift R a b w$ to
$\ilift S a b {(\unary i w)}$ and
$\binary{\quinary {\quinary \pliftoperator {R₁} {R₂} {R} {s₁} {s₂}}
  {t₁} {t₂} {f} {g} {w}} {w₁} {w₂}$ to
$\binary{\quinary {\quinary \pliftoperator {R₁} {R₂} {S} {s₁} {s₂}}
  {t₁} {t₂} {f} {g} {(i ∘ w)}} ({\binary \Sliftoperator i {w₁}})
({\binary \Sliftoperator i {w₂}})$.

The method described for $\Seqoperator$ can easily be extended to any
GADT. Indeed, given any GADT $G$, its completion $G_c$ is obtained by
first identifying each constructor $c : ∀ \bar {α}. Φ → \unary G {\bar
{Ψ}}$ whose return instance $\bar {Ψ}$ of $G$ is not simply $\bar {α}$
and replacing it by $c_c : ∀ \bar{α}\bar{β}. (\bar{Ψ → β}) → Φ
→ \unary G {\bar {β}}$. The name {\em completion} is justified by the
embedding $\iota_G$ of $G$ into $G_c$ defined on each element of the
form $\unary c x$ by $\binary {c_c} {\bar{\id{Ψ}}} x$. The completion
$G_c$ is akin to an ADT in the sense that the return type of each of
its constructors is a variable instance of $G_c$; however, it is not
an ADT {\em per se} because each of its constructors
$c_c:∀\bar{α}\bar{β}. (\bar{Ψ → β}) → Φ → \unary G {\bar {β}}$
quantifies over a possibly non-empty vector $\bar{α}$ of type
variables as well as over the vector $\bar{β}$ of type variables
appearing in $c_c$'s return instance. Nevertheless, the relational
interpretation $\hat{G_c}$ of $G_c$ can be constructed as in
Section~\ref{sec:inductive}, and the relational interpretation $\hat
G$ of $G$ is then defined by restriction as $\rel {\unary {\hat G} R}
g h = \rel {\unary {\hat{G_c}} R} {(\unary {\iota_G} g)} {(\unary
{\iota_G} h)}$. The resulting relational interpretation $\hat G$
inherits the inclusion-preservation property of $\hat{G_c}$ and also
ensures that each of the constructors $c$ of $G$ is parametric.

When $G$ is simply an ADT its completion $G_c$ is just a copy of $G$
itself (since then there is no constructor $c : ∀ \bar {α}. Φ → \unary
G {\bar {Ψ}}$ whose return instance $\bar {Ψ}$ is not $\bar {α}$). In
this case, the relational interpretation $\hat G$ defined in this
section is trivially the same as the relational interpretation
associated to it by the method described in
Section~\ref{sec:inductive}. The construction of $\hat G$ in this
section for an arbitrary GADT $G$ thus produces true generalizations
of the relational interpretations of ADTs, as introduced for lists by
Wadler and subsequently developed for more general data types by
others (see, e.g., \cite{bernardy,jg-nested}).


\section{Applications}
\label{sec:app}

\subsection{Free theorems}
\label{sec:free-thm}

The relational interpretations defined for GADTs in the previous
section can be used to establish free theorems à la Wadler. We
illustrate this with the data type $\Seqoperator$. Given a type $A$
and an element $a : A$, we say that a sequence $s : \Seq A$ {\em
contains only $a$ as data} either when $s = \inj a$ or when $s
= \pairing {s₁} {s₂}$ --- which forces $a$ to be of the form $\pair
{a₁} {a₂}$ --- and $s₁$ contains only $a₁$ as data and $s₂$ contains
only $a₂$ as data.

\begin{proposition}
  Let $f : ∀ α. α → \Seq {α}$ be a parametric polymorphic
  function. For any type $A$ and any element $a:A$, the sequence
  $\unary {\unary f A} a$ contains only $a$ as data.
\end{proposition}
\begin{proof}
  Since the function $f$ is parametric, for any types $A$ and $B$, any
  relation $R : \Rel A B$, any $a: A$ and any $b:B$, the sequences
  $\unary {\unary f A} a$ and $\unary {\unary f B} b$ are related by
  $\unary \Seqliftoperator R$ whenever $\rel R a b$ holds.  Now, fix a
  type $A$ and an element $a:A$, and consider the relation $δ_a : \Rel
  A A$ that relates $x$ and $y$ only when both are $a$ itself.
  Because $\rel {δ_a} a a$, we get a witness of $\Seqlift {δ_a}
  {(\unary {\unary f A} a)} {(\unary {\unary f A} a)}$. It remains to
  show, by induction on $s : \Seq A$, that if $\Seqlift {δ_a} s s$
  then $s$ contains only $a$ as data. If $s = \inj x$ for some $x :
  A$, then the element we have in $\Seqlift {δ_a} s s = \Slift {δ_a}
  {(\icons x)} {(\icons x)}$ must be of the form $\ilift {δ_a} x x
  {w}$ with $w : \rel{δ_a} x x$. By definition of $δ_a$, $w$ entails
  that $x$ is $a$ itself, and thus $s = \inj a$ does indeed contain
  only $a$ as data. If $s = \pairing {s₁} {s₂}$, then $A$ is of the
  form $\product {A₁} {A₂}$, $a$ is of the form $\pair {a₁} {a₂}$, and
  the element we have in $\Seqlift {δ_a} s s = \Slift {δ_a} {(\pcons
  {\id {\product {A₁} {A₂}}} {(\unary {ι} s₁)} {(\unary {ι} s₂)})}
  {(\pcons {\id {\product {A₁} {A₂}}} {(\unary {ι} s₁)} {(\unary {ι}
  s₂)})}$ must be of the form $\binary{\quinary
  {\quinary \pliftoperator {R₁} {R₂} {δ_a} {s₁} {s₂}} {s₁} {s₂} {\id
  {\product {A₁} {A₂}}} {\id {\product {A₁} {A₂}}} {w}} {w₁} {w₂}$
  with $R₁:\Rel {A₁} {A₁}$, $R₂ : \Rel {A₂} {A₂}$, $w
  : \rel{(\arrowlift{(\productlift {R₁} {R₂})} {δ_a})} {\id {\product {A₁}
  {A₂}}} {\id {\product {A₁} {A₂}}}$, $w₁ : \Seqlift {R₁} {s₁} {s₁}$,
  and $w₂ : \Seqlift {R₂} {s₂} {s₂}$. From $w$, we can prove that
  $R₁ \subseteq δ_{a₁}$ and $R₂ \subseteq δ_{a₂}$. Then from $w₁$ and
  $w₂$, inclusion-preservation of $\Seqliftoperator$ gives witnesses
  $v₁$ and $v₂$ of $\Seqlift {δ_{a₁}} {s₁} {s₁}$ and $\Seqlift
  {δ_{a₂}} {s₂} {s₂}$, respectively. By the induction hypothesis, $v₁$
  gives that $s₁$ has only $a₁$ as data and $v₂$ gives that $s₂$
  contains only $a₂$ as data. Thus, $s = \pairing {s₁} {s₂}$
  contains only $a = \pair {a₁} {a₂}$ as data.\,\looseness=-1
\end{proof}


\subsection{Graph lemma}
\label{sec:graph}

If $R$ is the graph of a function $f : A → B$, then $\unary {\hat G}
R$ is of particular interest. When $G$ is an ADT, then the graph lemma
says that $\unary{\hat G} R$ is exactly the graph of the function
$\mapfunap G f : \unary G A → \unary G B$, where $\mapfun G$ is the
usual map function associated with the ADT $G$. However, when $G$ is a
more general GADT, then there can be no standard map function
associated with $G$ (\cite{jp19,jc22}). In particular, the relation
$\unary{\hat{G}} R$ is not necessarily the graph of a function from
$\unary G A$ to $\unary G B$. Significantly, $\unary{\hat{G}} R$ can
still be understood as the graph of a {\em partial} function from
$\unary G A$ to $\unary G B$; see
Proposition~\ref{prop:graph-lemma-gadts} below. The key observation is
that the completion $G_c$ of $G$ has a map function $\mapfun {G_c} :
(A → B) → \unary{G_c} A → \unary{G_c} B$ whose application to $f$ is
defined for each constructor $c_c$ by $\mapfunaparg {G_c} f {(\binary
{c_c} h x)} = \binary {c_c} {(f\circ h)} x$.

\begin{proposition}
  \label{prop:graph-lemma-gadts}%
  Let $f : A → B$ be a function and let $G$ be a GADT. If $R$ is the
  graph of $f$, then $\unary{\hat G} R$ is the graph of a partial
  function, i.e., for any $x : \unary G A$, there is at most one
  $y : \unary G B$ such that $\ternary{\hat{G}} R x y$ is inhabited.
\end{proposition}
\begin{proof}
  Let $x: \unary G A$ and $y,y': \unary G B$ be such that
  $\ternary{\hat{G}} R x y$ and $\ternary{\hat{G}} R x y'$ are
  inhabited. Then
  $\ternary{\hat{G_c}} R {(\unary {ι} x)} {(\unary {ι} y)}$ and
  $\ternary{\hat{G_c}} R {(\unary {ι} x)} {(\unary {ι} {y'})}$ are
  inhabited as well. Since $G_c$ has a map function,
  $\unary{\hat{G_c}} R$ is also the graph of a function, so
  $\unary {ι} y = \unary {ι} {y'}$, and thus the injectivity of the
  embedding $ι$ implies $y = y'$.
\end{proof}

Using Proposition~\ref{prop:graph-lemma-gadts}, we can {\em define} a
``function mapping operation'' $\gmapfun G$ for $G$ by declaring
$\gmapfunaparg G f x$ to be $y$ if $\ternary{\hat{G}} R x y$ is
inhabited and undefined if no such $y$ exists. This definition of
$\gmapfun G$ generalizes the notion of mappability introduced
in~\cite{jc22}: when $f$ is mappable over $x$ in the specification $G$
in the sense of~\cite{jc22}, then the partial function $\gmapfunap G
f$ above is defined on $x$. The converse does not hold, however, as
shown in Example~\ref{ex:cex}.

\begin{example}
  \label{ex:cex} Consider the functions $f : \product\Bool\Bool
  → \product\Bool\Bool$ and $g: \Bool → \Bool$ defined by $\unary f
  {\pair x y} = \pair y x$ and $\unary g x = \neg x$. It should be
  fairly intuitive to the reader that we can map (for $\Seqoperator$)
  the function $\product f g$ over the sequence $s = \pairing {(\inj
  {\pair \true \false})} {(\inj \true)}$, with result $s' = \pairing
  {(\inj {\pair \false \true})} {(\inj \false)}$. More formally,
  according to the algorithm of~\cite{jc22}, $\product f g$ is indeed
  mappable over $s$ in the specification $\Seqoperator$, with result
  $s'$. In fact, writing $R$ for the graph of $\product f g$, we can
  give an actual witness of $\Seqlift R s {s'}$,
  namely\looseness=-1 \begin{displaymath} \binary{ \quinary
  {\quinary \pliftoperator {R₁} {R₂} {R} {(\icons
  {\pair \true \false})} {(\icons \true)}} {(\icons
  {\pair \false \true})} {(\icons \false)} {\id {}} {\id {}} {w}}
  {(\ilift {R₁} {\pair \true \false} {\pair \false \true} {w₁})}
  {(\ilift {R₂} \true \false {w₂})} \end{displaymath} Here, $R₁$ is
  the graph of $f$, $R₂$ is the graph of $g$, $w$ is a proof of
  inclusion (actually equality) of $\productlift {R₁} {R₂}$ in $R$,
  and $w₁$ and $w₂$ are witnesses of $\unary f {\pair \true \false}
  = \pair \false \true$ and $g \true = \false$, respectively. However,
  the partial function $\gmapfunap {\Seqoperator} (\product f g)$ is
  also defined on elements on which the algorithm of~\cite{jc22} would
  not consider $\product f g$ to be mappable, such as $t = \pairing
  {(\pairing {(\inj \true)} {(\inj \false)})} {(\inj \true)}$. Indeed,
  the algorithm finds $\product f g$ to not be mappable over $t$
  because $f$ is not of the form $\product {f₁} {f₂}$ but the first
  argument of the outer $\pairingoperator$ in $t$ is again constructed
  from $\pairingoperator$. Nevertheless, $\gmapfunaparg {\Seqoperator}
  {(\product f g)} t$ still exists and equals $t' = \pairing
  {(\pairing {(\inj \false)} {(\inj \true)})} {(\inj \false)}$ because
  there is a witness of $\Seqlift R t {t'}$,
  namely \begin{displaymath} \binary{ \quinary
  {\quinary \pliftoperator {R₁} {R₂} {R} {(\pcons {\id{}}
  {(\icons \true)} {(\icons \false)})} {(\icons \true)}} {(\pcons
  {\id{}} {(\icons \false)} {(\icons {\true})})} {(\icons {\false})}
  {\id {}} {\id {}} {w}} {w₁} {(\ilift {R₂} \true \false
  {w₂})} \end{displaymath} Here, $R₂$ is again the graph of $g$, and
  $w₂$ is again a witness of $g \true = \false$, but $R₁$ is the
  relation that only relates $\pair \true \false$ with
  $\pair \false \true$ (and nothing else), $w$ is a witness that
  $\productlift {R₁} {R₂}$ is included (strictly!) in $R$, and $w₁$ is
  of the form \begin{displaymath} \binary{ \quinary
  {\quinary \pliftoperator {R₁'} {R₁''} {R₁} {(\icons \true)}
  {(\icons \false)}} {(\icons \false)} {(\icons \true)} {\id {}} {\id
  {}} {z}} {(\ilift {R₁'} \true \false {w₁'})} {(\ilift
  {R₁''} \false \true {w₁''})} \end{displaymath} Here, $R₁'$ relates
  $\true$ with $\false$ (and nothing else) with witness $w₁'$, $R₁''$
  relates $\false$ with $\true$ (and nothing else) with witness
  $w₁''$, and $z$ is a witness of inclusion (actually equality) of
  $\productlift {R₁'} {R₁''}$ in $R₁$.
\end{example}


\section{Issues}
\label{sec:issues}

The research programme outlined above seems eminently
reasonable. However, to carry it out precisely, and thus to obtain
results such as those in Section~\ref{sec:app}, we need to define a
source language extending System F with built-in GADTs and a target
(dependent) type theory strong enough to express the relational
interpretations described in Sections~\ref{sec:inductive}
and~\ref{sec:gadts-comp}. In trying to do so, the following issues
concerning the target type theory arise:
\begin{enumerate}
\item\label{it:prrelations}
We have chosen to model relations between $A$ and $B$ as functions $A
  → B → \Set$, i.e., as proof-relevant relations. However, we have
  also freely used notions most naturally associated to
  proof-irrelevant relations, such as inclusion of relations. A
  rigorous treatment would either replace inclusion with an operation
  mapping witnesses to witnesses, or replace $\Set$ by a
  proof-irrelevant sort of propositions such as Coq's $\mathsf{Prop}$.
  The former choice amounts to representing relations as spans and
  replacing inclusions of relations by morphisms between those spans,
  whereas the latter choice provides direct support for the
  proof-irrelevant notions used herein.
\item\label{it:axiom-K} In Section~\ref{sec:graph}, we investigated a
  graph lemma for GADTs. There again, proof-relevance plays an
  important role. Indeed, the graph of a function $f: A → B$ is the
  relation $\Graph f : \Rel A B$ defined by
  $\Graphap f a b = (\unary f a ≡ b)$, where $≡$ is the equality type
  former of the prospective target type theory. Whether or not this
  relation $\Graph f$ is proof-relevant depends on whether or not the
  target type theory supports a version of Axiom K. Assuming Axiom K
  is, however, problematic: in a language like Agda, for example,
  Axiom K makes it possible to prove that the data constructors of
  inductive data types are injective. This is in direct opposition to
  the data constructors' expected parametric behavior.
\item\label{it:impredicative} To interpret GADTs correctly, the target
  type theory needs an impredicative universe that supports inductive
  constructions. The Calculus of Inductive Constructions offers such a
  universe (namely, $\mathsf{Prop}$) at the bottom of its
  hierarchy. However, this universe is proof-irrelevant, and thus is
  not suitable for our purposes. Indeed, carrying out our
  constructions in such a universe would effectively identify all data
  structures of any given instance of a GADT. In addition, it is
  well-known that impredicativity is inconsistent with strong
  dependent sums, which eliminates some obvious candidates for the
  target type theory.\looseness=-1
\end{enumerate}
Resolving Point~\eqref{it:prrelations} may require new technical
ideas, but we do not expect it to pose fundamental
difficulties. Points~\eqref{it:axiom-K} and~\eqref{it:impredicative}
pose a different challenge: either we design a target type theory with
the desired features, or we prove that we cannot. The former would
provide a framework for understanding parametricity of (languages
with) GADTs. The latter would definitively show that GADTs understood
as types of {\em data structures} cannot be parametric, and thus the
claim that GADTs {\em generalize} ADTs (implicit in the ``GADT''
terminology) is not justified.


\vspace*{0.1in}

\noindent%
{\bf Acknowledgment.} This work was supported by NSF awards CCF1906388
and CCF2203217.

\bibliographystyle{./entics}
\bibliography{bib}

\end{document}